\documentclass [12pt]{article}
\usepackage{amsmath}
\usepackage{amssymb}
\usepackage{amsthm}
\usepackage{amscd}
\usepackage{amsthm}
\usepackage{amsopn}
\usepackage{xspace}
\usepackage{verbatim}
\usepackage{mathscinet}
\usepackage{amsfonts}
\usepackage{epic}
\usepackage{eepic}
\usepackage[active]{srcltx}
\usepackage{euscript}
\usepackage{epsfig}
\usepackage{color}

\definecolor{red}{rgb}{1,0,0}
\definecolor{green}{rgb}{0,1,0}
\definecolor{SeaGreen}{RGB}{46,139,87}
\definecolor{Maroon}{RGB}{128,0,0}

\newcommand{\N}{\mathbb{N}}

\newcommand{\C}{{\mathbb{C}}}
\newcommand{\R}{{\mathbb{R}}}

\newcommand{\A}{\mathcal A}

\def\Dg {{\mathcal D}}

\def\Hg {{\mathcal H}}

\newcommand{\LL}{\mathcal L}
\def\Lg {{\mathfrak L}}

\def\PP{\mathcal P}

\def\Ug {{\mathcal U}}

\def\Vg {{\mathcal V}}
\def\Tg {{\mathcal T}}

\def\curl{\text{\rm curl}}

\def\curl{\text{\rm curl\,}}

\def\sign{\text{\rm sign\,}}

\newcommand {\pa}{\partial}

\newcommand {\ar}{\to}

\def\0{\mathbf  0}

\def\XXint#1#2#3{{\setbox0=\hbox{$#1{#2#3}{\int}$ }
\vcenter{\hbox{$#2#3$ }}\kern-.6\wd0}}

 

 \numberwithin{equation}{section}
\theoremstyle{plain}
\newtheorem{theorem}{Theorem}[section]
\newtheorem{lemma}[theorem]{Lemma}
\newtheorem{proposition}[theorem]{Proposition}

\newtheorem{remark}[theorem]{Remark}
\newtheorem{corollary}[theorem]{Corollary}

\begin{document}
\bibliographystyle{siam}
\title{On the spectrum of non-selfadjoint
  Schr\"odinger operators with compact resolvent}
\author{Y. Almog and B. Helffer}
 \date{}
 \maketitle
 \begin{abstract}
   We determine the Schatten class for the compact resolvent of Dirichlet realizations, in
   unbounded domains, of a class of non-selfadjoint differential
   operators. This class consists of operators
   that can be obtained via analytic dilation from a Schr\"odinger
   operator with magnetic field and a complex electric  potential. As
   an application, we prove, in a variety of examples motivated by Physics,  that
   the system of generalized eigenfunctions associated with the
   operator is complete, or at least the existence of an infinite
   discrete spectrum.
 \end{abstract}
\section{Introduction}
\label{sec:1}
The theory of non-selfadjoint differential operators is at a much less
developed state than that of selfadjoint theory. The lack of
variational methods makes it difficult, in many interesting cases, to
determine whether a non-selfadjoint operator $\PP$ possesses a
complete systems of generalized eigenvectors (by which we mean that
the vector space they span is dense), or even if the spectrum is
non-empty.  In addition, the definition of a closed extension of the
differential operator, is not always a straightforward matter. (There
are, of course, other questions of interest, such as the effectiveness
of the Fourier expansion \cite{he13,da00} which we do not address
here.) There is, however, significant interest in these questions (cf.
for instance \cite{da02,he13a,sikr12,A,trem05} to name just a few
references).

In recent contributions we have considered, together with X.B. Pan,
similar questions for a well defined closed extension
$\overline{\mathcal A}$ (or $\mathcal A^D$) of the differential
operator
\begin{equation}
\label{eq:17}
  \A = -\Big(\nabla-i\frac{x^2}{2}\hat{\i}_y)^2 + icy \,,
\end{equation}
where $\hat{\i}_y$ is a unit vector in the $y$ direction. We studied the
spectrum of this extension both in the entire $(x,y)$ plane \cite{AHP1}, where
we show that $\sigma(\overline{\A})=\emptyset\,$, and for the Dirichlet realization $\mathcal A^D$ of $\A$ in
the half plane
\begin{displaymath}
  \R^2_+ = \{ (x,y)\in\R^2 \,| \, y>0\,\} \,.
\end{displaymath}
In the latter case we show that the spectrum is not empty in the limit
$c\to0$ \cite{AHP3}, and in the limit $c\to\infty$ \cite{AHP2}, where our
techniques involve analytic dilation. The existence of a non-empty
spectrum for general values of $c$ remains an open  question.

In another contribution \cite{AH} we show that the normal state
for a superconductor in the presence of an electric current, and the
magnetic field it induces becomes locally unstable, under some
additional conditions which are omitted here, whenever $(\A-\lambda)^{-1}$
becomes unbounded, where $\lambda$ is a function of the electric current.
The analysis in \cite{AH} rests on some assumptions on the normal
electric potential, that are not the most general ones. Thus, for
instance, it assumes a non vanishing electric potential gradient. If
the boundary conditions are such that this assumption is violated,
one has to analyze the spectrum of different linear differential
operators instead of \eqref{eq:17}, that are still of some physical
interest. 

Let then $\Omega\subseteq\R^d$ be (possibly) unbounded, $\partial\Omega\in C^{2,\alpha}$ for
some $\alpha>0$, $A= (A_1,\dots, A_d) \in C^2(\bar{\Omega},\R^d)$.  Let $B(x)$ denote the
anti-symmetric matrix associated with $\curl A$:
$$
B_{jk} = \pa_k A_j - \pa_j A_k\,.
$$
We attempt to prove existence of a non-empty spectrum, and when
possible, to prove completeness of the system of (generalized)
eigenvectors of the operator $-\Delta_A+V$ in $L^2(\Omega,\C)$, where $V$ may be
complex valued. Since in some of the examples below we use analytic
dilation, we consider a more general class of operators.  (Note that
analytical dilation can be applied, in general, only in domains that
are invariant under real dilation).  In particular, we consider here
the operator
\begin{equation} \label{modele2a}
  \A_\alpha= -e^{-i \frac{\pi}{m(k+1)}} \frac{\partial^2}{\partial x^2} -
  e^{i \frac{\pi}{2 (k+1)}}\Big(\frac{\partial}{\partial y}-i\frac{x^m}{m}\Big)^2 +  e^{i \frac{\pi}{2 (k+1)}} y^{2k} \,,
\end{equation}
obtained via analytic dilation of the operator:
\begin{equation} \label{modele2aa}
  \A_0= - \frac{\partial^2}{\partial x^2} -
  \Big(\frac{\partial}{\partial y}-i\frac{x^m}{m}\Big)^2 + i  y^{2k} \,.
\end{equation}

We begin by defining the class of operators considered in this work. Let
then
$$(\alpha_1,\ldots,\alpha_d)\in(-\frac \pi 4, \frac \pi 4)^d$$ and
\begin{equation}
\label{eq:3}
  K\overset{def}{=}\min_{1\leq k\leq d} \cos(2\alpha_k) > 0 \,.
\end{equation}
Suppose that 
\begin{displaymath}
  V= V_1+V_2
\end{displaymath}
where $V_1\in C^1(\overline{\Omega},\C)$ satisfies
\begin{equation}
  \label{eq:1}
\Re V_1 \geq -\lambda^* \,,
\end{equation}
for some $\lambda^*\in\R$. \\
Suppose further that there exists a constant $C$ such that
\begin{equation}
\label{eq:2}
  |\nabla V_1|+ \max_{(k,l)\in\{1,\ldots,d\}^2} |\nabla B_{k\ell}|
  \leq C\, m_{B,V_1}  \text{ in } \Omega \,,
\end{equation}
where
\begin{equation}
  m_{B,V_1} = \sqrt{|V_1|^2 +|B|^2+ 1}\,,
\end{equation}
Assume in addition that $V_2\in L^\infty_{loc}(\overline{\Omega},\C)$ is such that for every $\epsilon>0$, there exists $C_\epsilon$ for
which
\begin{equation}
\label{eq:4}
  |V_2| 
  \leq C_\epsilon + \epsilon  \,  m_{B,V_1}   \text{ in } \Omega \,.
\end{equation}
Finally, to assure compactness of the resolvent, we assume that
\begin{equation}
  \label{eq:5}
m_{B,V_1} \xrightarrow[|x|\to +\infty]{} +\infty \text{ in } \Omega\,.
\end{equation}
Consider then the operator initially defined on $C_0^\infty(\Omega)$ by:
  \begin{equation}
\label{eq:6}
  \PP = \PP_0 + V\,,
\end{equation}
where
\begin{equation}\label{defpp0} 
  \PP_0 =- \sum_{k=1}^d e^{2i\alpha_k}\partial_{A_k}^2
\end{equation}
and
\begin{displaymath}
\partial_{A_k}: =\frac{\partial}{\partial x_k} -iA_k\,.
\end{displaymath}
We consider here the Dirichlet realization $\PP^D$
  of $\PP$ in 
  $\Omega$, i.e.,  some closed extension of $\PP$ which should be defined
  properly  on a subspace of functions satisfying a Dirichlet
  condition on $\partial\Omega$, or on the entire space $\R^d$. Our
  definition, based on some generalization of the Friedrichs extension
  in a non necessarily coercive case, will coincide with the standard
  notion when $\Omega$ is bounded,  $\Omega=\mathbb R^d$ or when the
  operator is selfadjoint and semi-bounded. 

Before stating our main result, we recall that, if $\Hg$ is a Hilbert
space and $p>0$  the Schatten class  $C^p(H)$ denotes the set of compact
operators $\Tg$ such that
\begin{displaymath}
 \|\Tg\|_p= \Big(\sum_{n=1}^\infty \mu_n(\Tg)^p\Big)^{1/p}<\infty\,,
\end{displaymath}
where $\mu_n(\Tg)$ are the eigenvalues of $(\Tg^*\Tg)^{1/2}$ repeated
according to their multiplicity \cite{dusc63,goetal90}. For $1\leq p$, $C^p$
is a Banach space with $\|\, \cdot \, \|_p$ as its norm. For $0<p<1$
$C^p$ is still well defined, but is not a Banach space and $\|\textperiodcentered{} \|_p$
is not a norm.

Our main result follows
\begin{theorem}
\label{thm:1.1}
  Under the above assumptions we have, for the Dirichlet realizations
  of $\PP$ and $-\Delta_A+|V|$, that for every
  $\lambda\in\rho(-\Delta_A+|V|)\cap\rho(\PP)$
 \begin{equation}\label{eq:18}
(-\Delta_A+|V|-\lambda)^{-1}\in C^p(L^2(\Omega,\C)) \Rightarrow (\PP-\lambda)^{-1}\in C^p(L^2(\Omega,\C))  \,.
  \end{equation}
  \end{theorem}

  The optimality of this result is also of interest. In this direction
  we prove the following
  \begin{theorem}
  With the notation and assumptions of the previous theorem,
 \begin{itemize}
  \item When $\Omega$ is either the entire space or the half-space in
    $\R^d$ we have
\begin{equation}
  (\PP-\lambda)^{-1}\in C^p(L^2(\Omega,\C)) \Rightarrow (-\Delta_A+|V|-\lambda)^{-1}\in C^p(L^2(\Omega,\C))
\end{equation}
 for every   $\lambda\in\rho(-\Delta_A+|V|)\cap\rho(\PP)$.
\item For a general domain $\Omega$ and for any $\alpha\in[-\pi,\pi)$, the Dirichlet realizations
  of  $-e^{i\alpha}\Delta_A+V$ and $-\Delta_A+|V|$ satisfy
  \begin{equation}
  (-e^{i\alpha}\Delta_A+V-\lambda)^{-1}\in C^p(L^2(\Omega,\C)) \Rightarrow (-\Delta_A+|V|-\lambda)^{-1}\in C^p(L^2(\Omega,\C))\,,
\end{equation}
for every
  $\lambda\in\rho(-\Delta_A+|V|)\cap\rho(-e^{i\alpha}\Delta_A+V)$.
\end{itemize}
\end{theorem}
In this last case, we consequently have  an equivalence, in the sense of
Schatten classes of the resolvents, between the Dirichlet realizations of $\PP$ and
$-\Delta_A+|V|$. \\
The complementary question is, naturally,  to which Schatten
class does the resolvent of $-\Delta_A+|V|$ belong? The following
theorem provides a satisfactory answer.
\begin{theorem}
\label{prop:1.2}
 Suppose that for some $p>0$
 \begin{equation}
\label{eq:26}
    \int_{\Omega\times  \R^d} (|\xi|^2 + m_{B,V_1})^{-p } dx d\xi <\infty \,.
 \end{equation}
Then, $ (-\Delta_A+|V|+1)^{-1}\in C^p(L^2(\Omega,\C))$.
\end{theorem}

Once the Schatten class for a compact operator has been obtained, one
can use the following fundamental result in operator theory to prove
completeness of its system of eigenvectors  (cf. for instance Theorem X.3.1 in
\cite{goetal90}, or Corollary XI.9.31 in \cite{dusc63}).
\begin{theorem}
\label{thm:1.4}
 Let $\Hg$ denote a Hilbert space, and $\A\in C^p(\Hg)$ be a compact operator 
for some $p > 0$. Assume that its numerical range
\begin{displaymath}
  W_{\A} = \{ \langle\A\varphi,\varphi\rangle \,| \, \varphi\in\Hg \,, \|\varphi\|=1\,\} 
\end{displaymath}
lies inside a closed angle with vertex at zero and opening $\pi/p\,$. \\
Let ${\rm Span}\,(\A)$ denote the closure of the vector space
generated by the generalized eigenfunctions. Then, ${\rm Span}\,(\A)=\Hg$ 
is complete.
\end{theorem}

We apply this statement to the resolvent of $\mathcal P$. Note that
\begin{multline*}
  W_{(\PP-\lambda)^{-1}} = \{ \langle(\PP-\lambda)^{-1}\varphi,\varphi\rangle \,| \, \varphi\in\Hg \,,
  \|\varphi\|=1\,\} \\ = \{  \langle\psi, (\PP-\lambda) \psi\rangle \,| \, \psi\in D(\PP)\,,
  \|(\PP-\lambda)\psi \|=1\,\} \,.
\end{multline*}
Hence, if for some $\lambda\in\rho(\PP)$, $W_{\PP-\lambda}$  lies in a closed angle
with vertex at zero and opening $\pi/p$, then so does $W_{(\PP-\lambda)^{-1}}$ and  it would follow immediately
that $\PP$ has a complete system of eigenfunctions.

The rest of the contribution is arranged as follows.\\
In the next section, we define the Dirichlet realization of $\PP$
and prove that its resolvent is compact.\\
In Section \ref{sec:3}  we prove Theorems
\ref{thm:1.1}-\ref{prop:1.2}. \\
Finally, in the last section, we use these results, together with
Theorem~\ref{thm:1.4} to prove completeness of the system of
(generalized)-eigenvectors, or at least existence of a non-empty
spectrum, for a few particular cases of \eqref{eq:6} motivated by
superconductivity problems.

\section{Definition of the Dirichlet realization}
\label{sec:2}
\subsection{Preliminaries}
As $\PP$ is defined by \eqref{eq:6} for smooth functions only, we seek
a closed extension $\PP^D$  corresponding to its  Dirichlet realization. For unbounded domains the definition of 
this extension deserves special attention. We thus consider 
 the sesquilinear form
 \begin{equation}
 (u,v) \mapsto a(u,v):= \sum_{k=1}^d \exp\{- 2i\alpha_k\} \langle \partial_{A_k} u\,,\, \partial_{A_k} v\rangle_{L^2(\Omega)}
  + \int_\Omega V(x) u(x) \bar v (x) \,dx
 \end{equation}
  initially defined on $C_0^\infty (\Omega)\times C_0^\infty (\Omega)$.\\
As is the common practice in such cases, it is useful to consider
instead, for some, sufficiently large,  $\gamma >0$
\begin{equation}\label{defa}
\begin{array}{ll}
 (u,v) \mapsto a_\gamma (u,v)&:=\sum_k \exp \{- 2i \alpha_k\} \langle \partial_{A_k} u\,,\, \partial_{A_k}
 v\rangle_{L^2(\Omega)} \\ & \qquad   + \int_\Omega V(x) u(x) \bar v (x) \,dx + \gamma \int
 u \bar v dx\,, 
 \end{array}
  \end{equation}
 to assure some coercivity.

 The Friedrichs extension of \eqref{defa} is a continuous sesquilinear
 form on $\Vg\times \Vg$, where
 \begin{equation}\label{defvg}
  \Vg = \{ u\in H^1_{0,A}(\Omega ) \, | \, |V_1|^{1/2}u\in L^2(\Omega,\C) \} \,,
\end{equation}
and $H^1_{0,A}(\Omega )$ denotes the closure of $C_0^\infty(\Omega)$ in the magnetic Sobolev space 
\begin{equation}
H^1_A(\Omega) = \{ u\in L^2(\Omega)\,,\, \nabla_Au \in L^2(\Omega)\}\,.
\end{equation}
When $\Omega$ is bounded, $\alpha_k=0$ for all $1\leq k\leq d$, and $V$ is real, 
the Dirichlet realization of $\PP$ can be easily obtained by applying
the Lax-Milgram Theorem (and the Friedrichs extension construction). When $V$ is complex valued and $\Im V\geq 0$,
the same method prevails possible if we employ a minor generalization
of the Lax Milgram Theorem, where hermitianity for the sesquilinear
form is no longer assumed \cite{AHP2}. When $\Im V$ has no definite sign
and is not bounded by $\Re V$ or the magnetic field, a more elaborate
generalization of the Lax-Milgram Theorem is needed.  In particular,
it is a necessary to replace the standard requirement for $\mathcal
V$-ellipticity (or coercivity) of \eqref{defa} by a weaker one.  This
is the object of the next subsection.

\subsection{A generalized Lax-Milgram Theorem}
 Let $\Vg$ denote a Hilbert space. Consider a continuous sesquilinear
form $a$ defined on $\Vg\times \Vg$:
\begin{displaymath}
  (u,v)\mapsto a(u,v)\,.
\end{displaymath}
 Recall that for a sesquilinear form continuity means that
for some $C>0$
\begin{equation}\label{lm1}
|a (u,v) |\leq C\;  \|u\|_\Vg\,  \| v \|_\Vg\,,\, \forall u, v \in \Vg\,.
\end{equation}
We denote the associated linear map by $\A\in \Lg(\Vg)$, i.e.,
\begin{equation}\label{lm2}
a(u,v) = \langle \A u\,,\,v\rangle_\Vg\,.
\end{equation}

\begin{theorem} \label{LMn}~\\
Let $a$ be a continuous sesquilinear form on $\Vg\times \Vg$. If $a$ satisfies,
for some $\Phi_1, \Phi_2 \in \mathcal L (\Vg)$
\begin{equation}
\label{lm5n}
|a(u,u)| + |a(u, \Phi_1(u))| \geq \alpha\, \|u\|_\Vg^2\,,\, \forall u\in \Vg\,.
\end{equation}
\begin{equation}
\label{lm5ne}
|a(u,u)| + |a( \Phi_2(u),u)| \geq \alpha\, \|u\|_\Vg^2\,,\, \forall u\in \Vg\,.
\end{equation}
then $\A$, as defined in \eqref{lm2}, is a continous isomorphism from
$\Vg$ onto $\Vg$. Moreover $\A^{-1}$ is continuous.
\end{theorem}
\begin{proof}
We split the proof into two different steps.
\\

{\bf Step 1}: $\A$ is injective, and has a closed range.\\

Choose $u\in \Vg$, such that $\A u=0$. This implies
\begin{equation}\label{lm6n}
 \langle \A u , u\rangle =0 \mbox{  and } \langle \A u , \Phi_1 (u) \rangle =0\,.
 \end{equation}
 It, however, follows from \eqref{lm5n}  that 
\begin{displaymath}
|\langle \A u\,,\,u\rangle _\Vg|  + |\langle \A u\,,\,\Phi_1(u) \rangle _\Vg| \geq  \alpha \|u\|_\Vg^2\,,\, \forall u\in \Vg\,.
\end{displaymath}
Hence,
\begin{displaymath}
  (1+\|\Phi_1\|)\| \A u \|_\Vg\cdot \| u\|_\Vg \geq  \alpha\, \|u\|_\Vg^2\,,\, \forall u\in \Vg\,,
\end{displaymath}
and consequently, for some $\tilde{\alpha}>0$
\begin{equation}\label{lm7}
\| \A u \|_\Vg  \geq  \tilde{\alpha} \, \|u\|_\Vg \,,\, \forall u\in \Vg\,,
\end{equation}
from which injectivity readily follows. Closedness of the range easily
follows from \eqref{lm7} and the continuity of $\A\,$.

{\bf Step 2}: $\A(\Vg)$ is dense in $\Vg$, and $\A^{-1}$ is continuous.\\

Consider $u\in \Vg$ such that $\langle \A v\,,\,u \rangle_\Vg=0\,,\, \forall v\in \Vg$. In
particular, we can choose $v_1=u$ and $v_2=\Phi_2(u)$ to obtain
$a(u,u)=0$ and $a(\Phi_2(u),u)=0$. Hence, by (\ref{lm5ne}) we must have
$u=0$. Thus, $\A$ is a bijection, $\A^{-1}:\Vg\to\Vg$ exists and is continous
by \eqref{lm7}.
\end{proof}

We now consider two Hilbert spaces $\Vg$ and $\Hg$ such that $\Vg\subset
\Hg$, and that for some $C>0$ and any $u\in V$, we have
\begin{equation}
  \label{eq:7}
 \|u\|_\Hg \leq C \|u\|_{\Vg}\,.
\end{equation}
Suppose further that
\begin{equation}\label{densi}
\Vg \mbox{ is dense in } \Hg\,.
\end{equation}
Let
\begin{equation}\label{lm3}
D(S) =\{u\in \Vg\;|\; v\mapsto a(u,v) \mbox{ is continuous on } \Vg \mbox{ in the
  norm of } \Hg\}.
\end{equation}
We can now define the operator $S:D(S)\to\Hg$ by
\begin{equation}
\label{lm4}
a(u,v) = \langle Su\,,\,v \rangle_\Hg\,,\, \forall u \in D(S) \mbox{ and } \forall v\in V\,.
\end{equation}

We can now prove
\begin{theorem}
\label{LaxMilgramv2}
  Let $a$ be a continuous sesquilinear form satisfying \eqref{lm5n}
  and \eqref{lm5ne}. Suppose, in addition, that $\Vg \subset \Hg$ and that
\eqref{eq:7} and \eqref{densi} hold. Assume further that
  $\Phi_1$ extends into a continuous linear map in $\mathcal L (\Hg)\,$.
Let $S$ be defined by \eqref{lm3}-\eqref{lm4}. Then
\begin{enumerate}
\item $S$ is bijective from $D(S)$ onto $\Hg$ and $S^{-1}\in \Lg(\Hg)\,$.
\item  $D(S)$ is dense in both $\Vg$ and $\Hg$
\item $S$ is closed.
\item Let $b$ denote the conjugate sesquilinear form of $a$, i.e.
  \begin{displaymath}
(u,v) \mapsto b(u,v):=\overline{a(v,u)}\,.
  \end{displaymath} 
Let $S_1$ denote the closed  linear operator associated with $b$ by the same construction . Then 
\begin{equation}\label{infoadj}
  S^*=S_1 \mbox{ and }S_1^* = S\,.
  \end{equation}
\end{enumerate}\
\end{theorem}

\begin{proof}
  We show here only that $S$ is injective. This is a consequence, for
all $u\in D(S)$\,, of
$$
\begin{array}{ll}
\alpha \|u\|^2_\Hg& \leq C\,   \alpha \|u\|^2_V\\
& \leq C \left( |a(u,u)| + | a(u, \Phi_1(u)) \right)|\\
& = C \, |
\langle S u\,,\,u \rangle_\Hg |   +| \langle S u\,,\,\Phi_1(u) \rangle_\Hg |\\
&    \leq
\hat C\,  \|Su\|_\Hg \,\cdot\, \|u\|_\Hg\,,\, 
\end{array}
$$
which leads to
\begin{equation}\label{p1}
\alpha \|u\|_\Hg \leq C\, \| Su\|_\Hg\,,\, \forall u\in D(S)\,,
\end{equation}
and
\begin{equation}\label{p2}
\alpha \|u\|_\Vg \leq C\, \| Su\|_\Hg\, \|u\|_\Hg\,,\, \forall u\in D(S)\,.
\end{equation}
Injectivity easily follows. 

We omit the rest of the proof, as it does not deviate from the proof
of the standard Lax-Milgram Theorem. Interested readers may find a presentation of the
missing details of the standard case in \cite{hel13}
\end{proof}
 \subsection{Definition of $\PP^D$}
We return to the operator $\PP$ introduced in (\ref{eq:6}) on
$C_0^\infty(\Omega)$  and describe how the previous abstract theory applies to
the construction $\PP^D\,$.
\\
\subsubsection{The case $V_2=0$}
 Let
$\Hg=L^2(\Omega,\C)$ and $\Vg$ as introduced in \eqref{defvg}. 
Initially, we equip $\Vg$ with the norm:
$$
u \mapsto \|u\|_{\Vg}:= \sqrt{ \|u\|^2_{H^1_A} + \int  \sqrt{|V_1|^2 +1}\,|u(x)|^2 \,dx}\,.
$$
We later prove (see \eqref{equiv}) that
$$
u \mapsto \|u \|_{\Vg,B}:= \sqrt{ \|u\|^2_{H^1_A} + \int m_{B,V_1} \, |u(x)|^2 \, dx}\,,
$$
is an equivalent norm on $\Vg$.\\
To apply the previous results to the sesquilinear form introduced in \eqref{defa}
we need to establish first that $a$ satisfies \eqref{lm5n} and
\eqref{lm5ne}. To this end we set $\Phi_1(u)=\Phi_2(u)=\phi_1u$, where
\begin{displaymath}
  \phi_1= \frac{\Im V_1} {m_{B,V_1}} \,.
\end{displaymath}
Clearly $\Phi_1$ belongs to $\Lg(\Vg)$ and $ \Lg(\Hg)$, since it is a
multiplication operator by a function in $W^{1,\infty}(\Omega)\,$.  Note that by
\eqref{eq:4} $\nabla\phi_1$ belongs to $ L^\infty (\Omega)\,$.

It can be easily verified that $a$ is continuous on $\Vg\times \Vg$. To use
Theorem~\ref{LaxMilgramv2}, we thus need to establish \eqref{lm5n} and
\eqref{lm5ne}. We first observe that for any $u\in \Vg\,$, we have:
\begin{equation}
\label{eq:9}
\Re a_\gamma (u,u) \geq  K  \int_\Omega |\nabla_A u(x)|^2 dx + \int (\Re V_1 (x)+\gamma) |u(x)|^2 dx\,,
\end{equation}
where $K$ is defined in \eqref{eq:3}. Furthermore,
\begin{displaymath}
\Im a (u, \phi_1 u) = \Im  \sum_{k=1}^d \int_\Omega e^{2i \alpha_k}
\partial_{A_k} u (x) \overline{\partial_{A_k} (\phi_1u)}\,   dx  
+ \int |\Im V_1|^2(m_{B,V_1})^{-1} |u|^2\, dx\,.
\end{displaymath}
After some simple manipulation we arrive at
\begin{multline*}
\Im a (u, \phi_1 u) = \Im  \sum_{k=1}^d  \int_\Omega e^{2i \alpha_k}
\partial_{A_k}u (x) \partial_k \phi_1\,\bar{u}\,dx  
 \\ + \int |\Im V_1|^2(m_{B,V_1})^{- 1} |u|^2\, dx
+  \sum_{k=1}^d  \int_\Omega  \phi_1 \sin (2\alpha_k) \, |\partial_{A_k}u|^2\, dx \,.
\end{multline*}
Clearly, as $\|\phi_1\|_\infty \leq 1\,$, we have
\begin{displaymath}
  \Big|\Im  \sum_{k=1}^d  \int_\Omega e^{2i \alpha_k} \partial_{A_k}u (x) \partial_k \phi_1\,\bar{u}\,dx
  \Big|\leq  \epsilon \, \|\nabla_Au\|_2^2 + C_\epsilon \,   \|u\|_2^2 \,.
\end{displaymath}
Consequently, there exists $\gamma_0$ and $C$ such that for  $\gamma>\gamma_0$
\begin{multline}
\label{eq:11}
C \left(\Im a_\gamma (u, \phi_1 u) + \Re a_\gamma) (u,u) \right)  \geq  \|\nabla_A u(x)\|^2_2 \,dx \\ +  \int \Re V_1
(x)|u(x)|^2 dx+   \int |\Im V_1|^2(m_{B,V_1})^{-1} |u|^2\, dx + \int |u(x)|^2\, dx \,.
\end{multline}

To complete the proof of \eqref{lm5n} we need an estimate for
$\|B\,m_{B,V}^{-1/2}u\|_2$. To this end we introduce the operator identity
($[\cdot,\cdot]$ being the Poisson bracket)
\begin{equation}
  \label{eq:17a}
B_{k\ell} = i [\partial_{A_k},\partial_{A_l}]\,.
\end{equation}
We then use \eqref{eq:17a}  to obtain
\begin{multline*}
 \int_{\Omega} B_{k\ell}(x)^2 (m_{\mathbf B,v})^{-1} | u |^2 \,dx \\
=i\int_{\Omega}  \{[\partial_{A_k},\partial_{A_l}] u\} \cdot  B_{k\ell}(x) (m_{B,V_1})^{-1}
\overline{u} \,dx\\ 
\leq C\Bigl\{ \|\partial_{A_l}u \|_2\|\partial_{A_k}u \|_2+ (\|\partial_{A_k}u \|_2 \\
 +\|\partial_{A_k}u\|_2) \|u \|_2
   \sup_{x\in\Omega}|\nabla (B_{k\ell} (m_{B,V_1})^{-1})|\Bigr\}.
\end{multline*}
As before this leads to
\begin{equation}
\label{2.18an}
\int_{\Omega} B_{k\ell}(x)^2 (m_{B,V})^{-1} |u|^2 \,dx
 \leq C (\|\nabla_A u \|^2_2 + \| u  \|^2_2)
\leq \tilde C   \,  \Re a (u,u)\,.
\end{equation}

We can now deduce from \eqref{2.18an} and \eqref{eq:11} that for some $C>0$,
\begin{equation}
\label{estima}
\| u\|_{\Vg,B}^2 \leq C\, ( | \Im a (u, \phi_1 u) | + | \Re a(u,u)|)\,,
\end{equation}
establishing both \eqref{lm5n} and \eqref{lm5ne} and also 
\begin{equation}\label{equiv}
\| u\|_{\Vg,B} \leq \widetilde C\,  \|u\|_{\Vg}\,.
\end{equation} 
Hence, the linear operator $S_\gamma$ associated with $a_\gamma$ can be defined over
the set \eqref{lm3}, and is an isomorphism from $D(S_\gamma)$ into
$L^2(\Omega,\C)$.  It can be easily verified that $a_\gamma (u,v)=\langle \PP_\gamma
u,v\rangle$ for all $u\in C_0^\infty(\Omega)$ (with $\PP_\gamma = \PP + \gamma)$ and hence
$S_{/C_0^\infty(\Omega)}= \PP_\gamma$. We can then define the extension of
$\mathcal P$ on $D(S_\gamma)$ by
\begin{displaymath}
\mathcal P^D = S_\gamma - \gamma I\,. 
\end{displaymath}
We have, hence, defined a closed extension of
$\PP$, on a set of functions satisfying, as $D(S) \subset \Vg$, a Dirichlet
boundary condition.

As a matter of fact, it can be easily verified that
\begin{equation}
  \label{eq:22}
D(\mathcal P^D)=D (S_\gamma) = \{ u \in \Vg\,|\,\PP u \in L^2(\Omega) \}
\end{equation}
where $ \PP u$ is defined as a distribution on $\Omega$.\\
Hence, $\PP^D$ defines the Dirichlet realization of $\PP$ in $\Omega$.

  Additionally, as $D(S_\gamma)\subset\Vg$ which is compactly embedded in
  $L^2(\Omega,\C)$ in view of \eqref{eq:5}, it follows that $\PP^D$ has
  a compact resolvent.   

\subsubsection{The general case}
We conclude this section by establishing the same results for
$\PP$ when $V_2$ is not necessarily $0$ but satisfies \eqref{eq:4}. 
We define to this end the sesquilinear form $b_\gamma\,:\,\Vg\times\Vg\to\C$
\begin{displaymath}
\begin{array}{ll}
  b_\gamma  (u,v) &:=\sum_k e^{2i \alpha_k} \int_\Omega \partial_{A_k}u(x) \overline{\partial_{A_k}v(x)}\, dx +
\int_\Omega (V(x)+\gamma) u(x)\,\bar v (x)\, dx \\
& = a_\gamma (u,v) + \int_\Omega V_2(x) u(x)\, \bar v(x) \, dx\,.
\end{array}
\end{displaymath}
In view of \eqref{eq:4}, $b_\gamma $ is continuous. Furthermore, for any
$u\in \Vg$ we have by \eqref{estima} that
\begin{displaymath}
\begin{array}{l}
  \|u\|^2_{H^1_A} + \int m_{B,V_1} |u(x)|^2 dx\\
   \quad \leq C(  \Im a_\gamma  (u, \phi_1 u)  +
   \Re a_\gamma (u,u)) \leq C( \Im b_\gamma (u, \phi_1 u)  +
  \Re b_\gamma (u,u) ) + C\langle|V_2|u,u\rangle \,.
  \end{array}
\end{displaymath}
We thus obtain, by \eqref{eq:4}, the existence of $\gamma_1$ and $C$ such that:
\begin{displaymath}
  \|u \|_{\Vg,B} \leq C( \Im b_\gamma (u, \phi_1 u) | +  \Re b_{\gamma + \gamma_1} (u,u))\,.
\end{displaymath}
It therefore follows, that we can apply to $b_{\gamma+\gamma_1}$ the same
  construction which was applied to $a_\gamma$, to obtain the same
  domain in the general case (and the same form domain) and that its
  resolvent is compact.

\section{The Schatten class} \label{sec:3}

In this section we attempt to obtain the optimal value of $p$ for the
Schatten class of the resolvent of the Dirichlet realization  $\PP^D$. We begin by showing that  if
$(-\Delta_A+|V|+1)^{-1}\in C^p$ then $(\PP^D-\lambda)^{-1}\in C^p$, thereby
allowing us to use techniques from selfadjoint theory. Then, we
provide a criterion on $V$ and $B$ which can be used to determine
whether the resolvent of $\PP^D$ is in a given Schatten class. For
convenience of notation we omit from now on the superscript $D$
and write  $\PP$ instead of $\PP^D$. 

\subsection{Comparison with a selfadjoint problem}

We begin with the following comparison result
\begin{proposition}
\label{lem:3.1}
Let $\{\mu_n\}_{n=1}^\infty$ denote the $n'th$ eigenvalue of
$(\PP^*\PP)^{1/2}$, where $\PP^*\PP$ is the linear operator associated
with the sesquilinear form $q:\,D(\PP)\times D(\PP)\to\C$ given by
  \begin{displaymath}
    q(u,v) = \langle\PP u,\PP v\rangle \,.
  \end{displaymath}
The domain of $\PP^*\PP$ is given by
\begin{displaymath}
  D(\PP^*\PP) = \{u\in D(\PP) \, | \, \PP  u\in D(\PP^*) \}\,.
\end{displaymath}
Consider then the Dirichlet  realization in $\Omega$ of $-\Delta_A+ m_{B,V_1}
$, and let
\begin{displaymath}
  \sigma(-\Delta_A+ m_{B,V_1} )=\{\nu_j\}_{j=1}^\infty \,.
\end{displaymath}
 Then, there exists $C>0$ such that
\begin{equation}
\label{eq:25}
  \nu_n \leq C\,(1+\mu_n) \,.
\end{equation}
\end{proposition}
\begin{proof}
For any $u\in\Vg$ we have
\begin{equation}
\langle ( - \Delta_A + m_{B,V_1} ) u,u \rangle \leq   C\, \left( | b (u, \phi_1 u) | + |
  b(u,u)| + \|u\|_2^2\right)\,.
\end{equation}
Hence, for any $u\in D(\PP)\,$, we can write
\begin{equation}
\langle ( - \Delta_A + m_{B,V_1} ) u,u \rangle \leq   C\, \left( | \langle \PP u, \phi_1 u \rangle | + | \langle \PP
  u , u\rangle |+ \|u\|_2^2\right)\,.
\end{equation}
Consequently, for all $u\in D(\PP)$,
\begin{equation}
\langle ( - \Delta_A + m_{B,V_1} ) u,u \rangle \leq   C\, \left( \| \PP u\|_2 \| u\|_2 + \|u\|_2^2\right)\,.
\end{equation}

It thus follows that for each $j\geq1$ there exists
a vector space $E_j$ of dimension $j$ in $D(\PP^*\PP)$ (hence 
also in $D(\PP)$) such that, for all $u
\in E_j$, we have
\begin{equation}
\langle ( - \Delta_A + m_{B,V_1} ) u,u \rangle \leq   C\, (\mu_j+1)\,  \| u\|^2_2\,.
\end{equation}
By Proposition 11.9 in \cite{hel13} applied to the Dirichlet
realization of \break $ ( - \Delta_A + m_{B,V_1} )$ (observing that the
domain of the operator can be replaced by the form domain $\Vg$ of
this Dirichlet realization and that $D(\PP)\subset\Vg$), we then obtain
\eqref{eq:25}.
\end{proof}

\subsection{Proof of Theorem \ref{prop:1.2}}
By \eqref{equiv} we have that
\begin{displaymath}
  \langle u,(-\Delta_A+m_{B,V_1})u\rangle\leq  \tilde{C} \langle u,(-\Delta_A+|V|+1)u\rangle\,.
\end{displaymath}
Thus, by the min-max principle, the resolvents of the operators $-\Delta_A+m_{B,V_1}$ and
$-\Delta_A+|V|+1$ always belong to the same Schatten class.  We
therefore  begin by restating the theorem in the following equivalent
form
\begin{theorem}
\label{thschatten}
 Suppose that 
 \begin{equation}
\label{eq:26}
    \int_{\Omega\times  \R^d} (|\xi|^2 + m_{B,V_1} )^{-p }\, dx d\xi <\infty \,.
 \end{equation}
Then, $(-\Delta_A+m_{B,V_1})^{-1} \in C^p(L^2(\Omega,\C))\,$.
\end{theorem}
We note that by \eqref{eq:25} it follows that whenever \eqref{eq:26}
is satisfied, then $(\PP-\lambda)^{-1}\in C^p$.
\begin{proof}
For  $\Omega=\mathbb R^d$, \eqref{eq:26} has been established in
\cite{CSS}. To extend it to the Dirichlet realization of
$-\Delta_A+m_{B,V_1}$ for general domains we extend $$U=m_{B,V_1}$$  to
$\R^d$ in the following manner: 
\begin{displaymath}
  U_{\rho,M} =
  \begin{cases}
  U & \text{in } \Omega \\
   \rho (1+x^2)^M & \text{in } \Omega^{c} \,,
  \end{cases}
\end{displaymath}
where $M$ is chosen so that
\begin{displaymath}
   \int_{\R^d \times \R^d} (|\xi|^2 +  (1+x^2)^{M})^{-p} \,  dx d\xi <\infty \,.
\end{displaymath}
and $\rho \geq 1\,$.\\
Let
\begin{displaymath}
  \sigma(-\Delta_A + U_{\rho,M}) = \{\mu_j\}_{j=1}^\infty \;,\; \sigma(-\Delta_A + U) =
  \{\lambda_j\}_{j=1}^\infty \,,
\end{displaymath}
where $-\Delta_A + U_{\rho,M}$ is the unique self-adjoint extension on $\mathbb R^d$ (by Kato's theorem) and $-\Delta_A + U$ is the Dirichlet realization in $\Omega$.\\
It can be easily verified  by comparison of the form domains that
\begin{displaymath}
  \mu_j \leq \lambda_j \,,\quad \forall j\geq1 \,.
\end{displaymath}
From Theorem 2.1 in \cite{CSS} we then get, for any $\rho \geq 1\,$, 
\begin{multline}
  {\rm Trace}  (-\Delta_A + U + 1)^{-p} \\ \leq  {\rm Trace}  (-\Delta_A +
U_{\rho,M}+1)^{-p} \\ \leq  (2\pi)^{-n} \int_{ \R^d\times \R^d} (\xi^2 + U_{\rho,M}+1)^{-p} dx d\xi\,.
\end{multline}
Taking the limit $\rho \to\infty$ yields
\begin{displaymath}
  {\rm Trace}  (-\Delta_A + U + 1)^{-p}  \leq  (2\pi)^{-n} \int_{ \Omega\times \R^d} (\xi^2
  + U +1)^{-p} dx d\xi\,.
\end{displaymath}
It follows that if \eqref{eq:26} holds true, then $(-\Delta_A +
U)^{-1}\in C^p(L^2(\Omega,\C))$.
\end{proof}

\subsection{On the optimality of the criterion \eqref{eq:26}} 
\label{sec:3.1}
Let
\begin{displaymath}
  \Dg= \{ u\in C^\infty(\bar{\Omega},\C)\cap H^1_0(\Omega,\C)\,|\, {\rm Supp}\,u \text{
    is a compact subset of } \overline{\Omega}\} \,.
\end{displaymath}
It is clear that $\Dg$ is in $D(\PP^D)$. 
We first show
\begin{lemma}
\label{lem:dens}
  Let $D(\PP^D)$ be given by \eqref{eq:22}. Then, $\Dg$ is
  dense in $D(\PP^D)$ under the  norm $\|\PP \cdot\|_2 + \|\cdot\|_{\Vg}$.
\end{lemma}
\begin{proof}
  Let $u\in D(\PP^D)$. By  \eqref{p2},  \eqref{lm5ne}  and \eqref{equiv} we have 
  \begin{equation}
\label{eq:41}
    \| u\|_{\Vg,B} \leq C\, (\|u\|_2 + \|\PP u\|_2)\,.
  \end{equation}
  Let $\eta\in C^\infty(\R_+,[0,1])$ denote a cutoff function
  satisfying
\begin{displaymath}
  \eta(t)=
  \begin{cases}
    1 & t<1 \\
    0 & t >2\,.
  \end{cases}
\end{displaymath}
Denote then by $\eta_k\,:\,\Omega\to\R_+$  the restriction to $\Omega$ of the
cutoff function defined by $\eta_k(x)=\eta(|x|/k)$ for all $x\in\Omega$. By
our assumption on $\partial\Omega$ we have that $\eta_k\in C^2(\overline{\Omega})$ for all $k\geq1$.

Next, define the sequence $\{v_k\}_{k=1}^\infty$ through $v_k=\eta_k u\,$.
Using the local regularity of the Dirichlet problem we easily conclude
that $v_k\in H^2(\Omega,\C)$ has compact support, and it can be readily
verified that $\|v_k - u\|_{\Vg} \ar 0$ as $k\ar
+\infty$.\\
We now prove that $v_k\to u$ in the graph norm. To prove that
\begin{equation}\label{limi}
\|\PP(v_k-u)\|_2\to0
\end{equation}
 we first compute
\begin{displaymath}
  \PP (v_k-u) = (\eta_k-1) \PP u +
  u\sum_{m=1}^de^{2i\alpha_m}\frac{\partial^2\eta_k}{\partial x_m^2} +
  2\sum_{m=1}^de^{2i\alpha_m}\frac{\partial\eta_k}{\partial x_m}\partial_{A_m}u \,.
\end{displaymath}
The first two terms on the right-hand-side tend to $0$ (in $L^2$
sense) since both $u$ and $\PP u$ are in $L^2(\Omega,\C)$. For the last
term we have
\begin{displaymath}
  \Big\|\sum_{m=1}^de^{2i\alpha_m}\frac{\partial\eta_k}{\partial x_m}\partial_{A_m}u \Big\|_2 \leq
  \frac{C}{k}\|\nabla_Au\|_2 \,.
\end{displaymath}
Hence, by \eqref{eq:41}, we obtain \eqref{limi}.

The result of the foregoing discussion is that
\begin{displaymath}
   \tilde{\Dg}= \{ u\in H^2(\bar{\Omega},\C)\cap H^1_0(\Omega,\C)\,|\, {\rm Supp}\,u \text{
    is a compact subset of } \overline{\Omega}\} \,,
\end{displaymath}
is dense in $D(\PP^D)$ with respect to the  norm introduced in the lemma. One can now
complete the proof invoking standard arguments that show that $\Dg$ is
dense in $\tilde{\Dg}$ under the same norm. 
\end{proof}
We continue this subsection by the following  lemma. 
\begin{lemma}
  \label{lem:aux} 
 Under Assumptions  (\ref{eq:1})-(\ref{defpp0}), there exists $C(\Omega,\PP)$ such that
for all $u\in\Dg$ we have
\begin{equation}
  \label{eq:19}
\|Bu\|_2^2 \leq C(\|\Delta_Au\|_2^2 + \|V_1 u \|_2^2+\|u\|_2^2) \,.
\end{equation}
\end{lemma}
\begin{proof}
  Let $u\in\Dg$. We use the identity \eqref{eq:17a} to obtain, after
  integration by parts 
   \begin{equation}
\label{eq:40}
     \|B_{km}u\|_2^2 =
     \Im\langle[\partial_{A_k},\partial_{A_m}]u,B_{km}u\rangle\leq  |\langle\partial_{A_m}u,\partial_{A_k}(B_{km}u)\rangle| +   |\langle\partial_{A_k}u,\partial_{A_m}(B_{km}u)\rangle|\,. 
   \end{equation}
As
\begin{displaymath}
  \langle\partial_{A_m}u\,,\,\partial_{A_k}(B_{km}u)\rangle =  \langle B_{km}\partial_{A_m}u\,,\,\partial_{A_k}\,u\rangle +
  \langle\partial_{A_m}u\,,\,u\partial_kB_{km}\rangle\,,
\end{displaymath}
it follows from \eqref{eq:2} and \eqref{eq:40} that
\begin{displaymath}
   \|B_{km}u\|_2^2 \leq
   2\||B_{km}|^{1/2}\partial_{A_m}u\|_2\||B_{km}|^{1/2}\partial_{A_k}u\|_2 +
  C\, \|m_{B,V_1}u\|_2 (\|\partial_{A_m}u\|_2 + \|\partial_{A_k}u\|_2 ) \,.
\end{displaymath}
Consequently, for any $\epsilon>0$,
\begin{multline*}
  \|B_{km}u\|_2^2  \leq
   2\||B_{km}|^{1/2}\partial_{A_m}u\|_2\||B_{km}|^{1/2}\partial_{A_k}u\|_2 \\
  + C\Big[\frac{\epsilon}{2}\|m_{B,V_1}u\|_2^2 +
   \frac{1}{4\epsilon}(\|\partial_{A_m}u\|_2^2+ \|\partial_{A_k}u\|_2^2+\|u\|_2^2)\Big] \,.
\end{multline*}
Hence, for all $\epsilon>0$ there exists $C_\epsilon>0$ such that
\begin{multline*}
  \|B_{km}u\|_2^2 - \epsilon\|Bu\|_2^2\\ \leq C_\epsilon \left(\||B_{km}|^{1/2}\partial_{A_m}u\|_2^2 +
  \||B_{km}|^{1/2}\partial_{A_k}u\|_2^2 +\|V_1u\|_2^2\right.\\
  \left. +   \|\partial_{A_m}u\|_2^2+ \|\partial_{A_k}u\|_2^2 +\|u\|_2^2\right) \,.
\end{multline*}
Summing over $k$ and $m$ we then obtain, using the standard inequality
\begin{equation}\label{ineqstand}
  \|\nabla_Au\|_2^2 \leq \frac{1}{2}(\|\Delta_Au\|_2^2 + \|u\|_2^2)\,,
\end{equation}
that for sufficiently small $\epsilon$ ,
\begin{displaymath}
\|Bu\|_2^2 \leq C ( \||B|^{1/2}\nabla_Au\|_2^2 + \|\Delta_Au\|_2^2+\|V_1u\|_2^2+\|u\|_2^2)\,,
\end{displaymath}
and hence,
\begin{equation}
\label{eq:43}
  \|m_{B,V_1}u\|_2^2 \leq C ( \||B|^{1/2}\nabla_Au\|_2^2 + \|\Delta_Au\|_2^2+\|V_1u\|_2^2+\|u\|_2^2)\,.
\end{equation}

Next, we use integration by parts to show that
\begin{displaymath}
  -\langle m_{B,V_1}u\,,\, \Delta_Au\rangle = 
  \langle m_{B,V_1}\nabla_Au\,,\, \nabla_Au\rangle +
  \langle u\nabla m_{B,V_1}\, ,\, \nabla_Au\rangle \,.
\end{displaymath}
By \eqref{eq:2} we have that
\begin{displaymath}
  |\nabla m_{B,V_1}| \leq C \, m_{B,V_1} \,.
\end{displaymath}
Consequently,
\begin{displaymath}
   - \langle m_{B,V_1}u\,,\, \Delta_Au\rangle \geq
   \|m_{B,V_1}^{1/2}\nabla_Au\|_2^2 -   C\, \|m_{B,V_1}u\|_2\, \|\nabla_Au\|_2 \,.
\end{displaymath}
With the aid of \eqref{ineqstand},  we then obtain that
\begin{displaymath}
  \|m_{B,V_1}u\|_2\|\Delta_Au\|_2 \geq \|m_{B,V_1}^{1/2}\nabla_Au\|_2^2- C \|m_{B,V_1}u\|_2\|(\|\Delta_Au\|_2 + \|u\|_2)\,.
\end{displaymath}
Hence,
\begin{displaymath}
  \|m_{B,V_1}^{1/2}\nabla_Au\|_2^2 \leq C\, \|m_{B,V_1}u\|_2(\|\Delta_Au\|_2 +\|u\|_2)\,,
\end{displaymath}
from which we conclude that
\begin{equation}
\label{eq:44}
  \|m_{B,V_1}^{1/2}\nabla_Au\|_2^2 \leq \epsilon\,  \|m_{B,V_1} u\|^2 +  C_\epsilon \, (\|\Delta_Au\|_2^2 +\|u\|_2^2 )\,,
\end{equation}
which, combined with \eqref{eq:43}, easily yields \eqref{eq:19}. 
\end{proof}
We continue with the following comparison  result
\begin{lemma}
\label{lem:opt}
Let $\A=-e^{i\alpha}\Delta_A+V$, and let $\{\mu_n\}_{n=1}^\infty$ denote the $n'th$
eigenvalue of $(\A^*\A)^{1/2}$.  Let further
\begin{displaymath}
  \sigma((-\Delta_A+|V|))=\{\nu_j\}_{j=1}^\infty \,.
\end{displaymath}
Then, there exists $C>0$ and
$n_0>0$, such that for all $n>n_0$ we have
\begin{equation}
\label{eq:14}
  \mu_n \leq C(1+\nu_n) \,.
\end{equation}
\end{lemma}
\begin{proof}
Clearly, for every $u\in\Dg$,
  \begin{equation}
\label{eq:23}
     \|(-\Delta_A+|V|)u\|_2^2 =  \|\Delta_Au\|_2^2 + \|Vu\|_2^2 - 2\Re\langle |V|u,\Delta_Au\rangle \,.
  \end{equation}
As
\begin{displaymath}
  2\Re\langle (|V_1+V_2|-|V_1|)u,\Delta_Au\rangle \leq  2\||V_2|u\|_2\|\Delta_Au\|_2\,,
\end{displaymath}
it follows, after integration by parts, that
\begin{equation}
\label{eq:20}
   2\Re\langle |V|u,\Delta_Au\rangle \leq   - 2\langle |V_1|\nabla_Au,\nabla_Au\rangle  + 2|\Re\langle u\nabla |V_1|,\nabla_Au\rangle| +
   2\||V_2|u\|_2\|\Delta_Au\|_2\,.
\end{equation}
For every $\epsilon>0$ we have, by \eqref{eq:4}, that
\begin{multline}
\label{eq:21}
 2 |\Re\langle u\nabla |V_1|,\nabla_Au\rangle \leq \epsilon\|u\nabla V_1\|_2^2 + \frac{1}{\epsilon}\|\nabla_Au\|_2^2\\
\leq 
 C\epsilon\, (\|V_1u\|_2^2 + \||B|u\|_2^2)- \frac{1}{\epsilon}\, \Re\langle u,\Delta_Au\rangle 
\,.
\end{multline}
By \eqref{eq:19}, for each
positive $\epsilon$, there exists $C_\epsilon$ such that
\begin{equation}
\label{eq:12}
  2 |\Re\langle u\nabla |V_1|,\nabla_Au\rangle| \leq \epsilon(\|V_1u\|_2^2+ \|\Delta_Au\|_2^2)+C_\epsilon \|u\|_2^2\,.
\end{equation}

For the last term on the right-hand-side we have, by \eqref{eq:4},
that for any $\epsilon>0$ there exists
$C_\epsilon>0$ for which
\begin{displaymath}
  2\|V_2u\|_2\|\Delta_Au\|_2 \leq \epsilon\, (\|V_1u\|_2^2+\||B|u\|_2^2+\|\Delta_Au\|_2^2) + C_\epsilon\|u\|_2^2 \,.
  \end{displaymath}
  With the aid of \eqref{eq:19}, we then conclude that for every $\epsilon>0$
  there exists $C_\epsilon>0$ such that
\begin{displaymath}
2\|V_2u\|_2\|\Delta_Au\|_2 \leq  \epsilon\, (\|\Delta_Au\|_2^2+\|Vu\|_2^2)+C_\epsilon\|u\|_2^2 \,.
\end{displaymath}
Substituting the above together with \eqref{eq:12} 
into \eqref{eq:20} yields that for every $\epsilon>0$ there exists $C_\epsilon>0$
such that
\begin{displaymath}
  2\Re\langle |V|u,\Delta_Au\rangle \leq \epsilon\, \big(\||V|u\|_2^2+ \|\Delta_Au\|_2^2\big) + C_\epsilon\|u\|_2^2\,.
\end{displaymath}
Substituting the above into \eqref{eq:23} we obtain, if we choose $\epsilon$ sufficiently
small, 
\begin{equation}
\label{eq:48}
   \|\Delta_Au\|_2^2 + \|Vu\|_2^2  \leq  C\, \|(-\Delta_A+|V|)u\|_2^2 \,.
\end{equation}
From here we easily obtain that there exists $C>0$ such that  for all
$u\in\Dg$ 
    \begin{equation}
\label{eq:24}
\|(-e^{i\alpha}\Delta_A+V)u\|_2^2 \leq C\, (\|(-\Delta_A+|V|)u\|_2^2 + \|u\|_2^2) \,.
  \end{equation}
We now use Lemma \ref{lem:dens} to conclude that $\Dg$ is dense in
$D(-\Delta_A+|V|)$. Hence for all $n\geq1$ there exist $\{u_k\}_{k=1}^n\subset\Dg$
such that
\begin{displaymath}
  \langle(-\Delta_A+|V|)u,u\rangle \leq 2\nu_n\|u\|_2^2 \;,\; \forall u\in{\rm
    span}\{u_k\}_{k=1}^n \,.
\end{displaymath}
We can now obtain \eqref{eq:14} by using once again Proposition 11.9
in \cite{hel13} (as in the proof of \eqref{eq:25}) together with
\eqref{eq:24} and the fact that $\Dg\subset D(-e^{i\alpha}\Delta_A+V)$.
\end{proof}

\begin{remark}
  The above lemma, together with \eqref{eq:25} proves that for
  any $p>0$ and $\lambda\in\rho(-e^{i\alpha}\Delta_A+V)\cap\rho(-\Delta_A+|V|)$,
  \begin{equation}
\label{eq:8}
    (-\Delta_A+|V|-\lambda)^{-1}\in C^p \Leftrightarrow(-e^{i\alpha}\Delta_A+V-\lambda)^{-1}\in C^p \,,
  \end{equation}
whereas Theorem \ref{thschatten} provides us only with an upper bound
for the optimal value of $p$. While the Dirichlet realization of
$-\Delta_A+|V|$ in $\Omega$ is self-adjoint, the authors are unaware of an
asymptotic expansion for the counting function associated with it in
the necessary generality. There is, however, reason to believe
that (\ref{eq:26}) is optimal. An example where such an asymptotic
expansion has been derived is given in \cite{L-BNo,leetal90} for the
case $\Omega=\R^d\,$, and where $A$ and $V$ are polynomials.
\end{remark}

We now prove  optimality in the general case for the particular case
when $\Omega$ is the half-space in $\R^d$.
\begin{lemma}
Let $\Omega=\R_+^d$ where
\begin{displaymath}
  \R_+^d = \{ x\in\R^d \,|\, x_d>0\,\} \,.
\end{displaymath}  
Let further 
\begin{displaymath}
  \Dg_1= \{ u\in\Dg,|\, \PP u\in\Dg \} \,.
\end{displaymath}
Then, there exists $C(\Omega,\PP)>0$, such that for all $u\in\Dg_1$
  we have
  \begin{equation}
 \label{eq:45}
\|\PP u\|_2^2 \leq C[\|(\Delta_A+|V|)u\|_2^2 + \|u\|_2^2] \,.
  \end{equation}
\end{lemma}
\begin{proof}
  Clearly,
  \begin{displaymath}
   \|\PP_0u\|_2^2 = \sum_{k,m=1}^d e^{2i(\alpha_k-\alpha_m)}
   \langle\partial_{A_m}^2u,\partial_{A_k}^2u\rangle    \,.
  \end{displaymath}
We now write
\begin{displaymath}
  \langle\partial_{A_m}^2u,\partial_{A_k}^2u\rangle = -
  \langle\partial_{A_k}\partial_{A_m}^2u,\partial_{A_k}u\rangle + \int_{\partial\Omega}
  \partial_{A_m}^2u\cdot \overline{\partial_{A_k}u}\,\nu_k \,ds\,.
\end{displaymath}
Since $u\in\Dg_1$, we have that $\PP_0 u=0$ on $\partial\Omega$, and hence
\begin{displaymath}
 \sum_{m=1}^d e^{2i(\alpha_k-\alpha_m)} \int_{\partial\Omega}
 \partial_{A_m}^2u \cdot \overline{ \partial_{A_k}u}\,\nu_k \,ds=0\,, \forall1\leq k\leq d \,.
\end{displaymath}
Consequently,
\begin{equation}
 \label{eq:46}
 \|\PP_0^2u\|_2^2 =- \sum_{k,m=1}^d e^{2i(\alpha_k-\alpha_m)}
 \langle\partial_{A_k}\partial_{A_m}^2u,\partial_{A_k}u\rangle \,.
\end{equation}

Next we write
\begin{multline}
\label{eq:38}
  \langle\partial_{A_k}\partial_{A_m}^2u,\partial_{A_k}u\rangle=
  \langle\partial_{A_m}\partial_{A_k}\partial_{A_m}u,\partial_{A_k}u\rangle + i\langle
  B_{km}\partial_{A_m}u,\partial_{A_k}u\rangle=\\  \langle\partial^2_{A_m}\partial_{A_k}u,\partial_{A_k}u\rangle + i(\langle
  B_{km}\partial_{A_m}u,\partial_{A_k}u\rangle+\langle
  \partial_{A_m}B_{km}u,\partial_{A_k}u\rangle)\,.
\end{multline}
Integration by parts yields
\begin{equation}
\label{eq:39}
  - \langle\partial^2_{A_m}\partial_{A_k}u,\partial_{A_k}u\rangle = \|\partial_{A_m}\partial_{A_k}u\|_2^2 +\int_{\partial\Omega}
  \nu_m \partial_{A_m}\partial_{A_k}u\,\overline{ \partial_{A_k}u} \,ds\,. 
\end{equation}
As $u=0$ on the boundary, the surface integral on the right-hand-side
vanishes whenever $(m,k)\neq(d,d)$. If $m=k=d$, we use the fact that
$\PP_0^2 u=0$ on the boundary to obtain
\begin{equation}
\label{eq:37}
  \int_{\partial\Omega}\partial_{A_d}\partial_{A_d}u \, \overline{\partial_{A_d}u}\,ds=
  \sum_{n=1}^{d-1}e^{2i(\alpha_d-\alpha_n)} \int_{\partial\Omega} \partial_{A_n}\partial_{A_n}u\, \overline{\partial_{A_d}u}\,ds=0\,.
\end{equation}
Combining the above with \eqref{eq:38}, \eqref{eq:39}, and
\eqref{eq:46}, yields
\begin{multline}
\label{eq:47}
  \|\PP_0u\|_2^2 = \Re\sum_{k,m=1}^d
  e^{2i(\alpha_k-\alpha_m)}\big[\|\partial_{A_m}\partial_{A_k}u\|_2^2 \\- i(\langle
  B_{km}\partial_{A_m}u,\partial_{A_k}u\rangle+\langle
  \partial_{A_m}B_{km}u,\partial_{A_k}u\rangle)\big]\,.
\end{multline}

We proceed by estimating the two rightmost terms in \eqref{eq:47}. To
this end we use \eqref{eq:44} to obtain that
\begin{equation}
\label{eq:10}
  |\langle B_{km}\partial_{A_m}u,\partial_{A_k}u\rangle| \leq  \||B|^{1/2}\nabla_Au\|_2^2 \leq
  C(\|\Delta_Au\|_2^2 + \|Vu\|_2^2 +\|u\|_2^2)\,.
\end{equation}
For the last term on the right-hand-side we have
\begin{displaymath}
 \langle\partial_{A_m}B_{km}u,\partial_{A_k}u\rangle = \langle B_{km}\partial_{A_m}u,\partial_{A_k}u\rangle +
 \langle u\partial_mB_{km},\partial_{A_k}u\rangle \,.
\end{displaymath}
The first term on the right-hand-side of the above equation has
already been estimated by \eqref{eq:10}. For the second term we use
\eqref{eq:2} and \eqref{eq:19} to obtain
\begin{equation}
\label{eq:34}
   |\langle u\partial_mB_{km},\partial_{A_k}u\rangle| \leq C\|m_{B,V}u\|_2\|\nabla_Au\|_2 \leq C(\|\Delta_Au\|_2^2 + \|Vu\|_2^2 +\|u\|_2^2)\,.
\end{equation}
Substituting the above together with \eqref{eq:10} into \eqref{eq:47}
yields
\begin{displaymath}
  \|\PP_0u\|_2^2 \leq  \Re\sum_{k,m=1}^d
  e^{2i(\alpha_k-\alpha_m)}\|\partial_{A_m}\partial_{A_k}u\|_2^2 +  C(\|\Delta_Au\|_2^2 + \|Vu\|_2^2 +\|u\|_2^2)\,.
\end{displaymath}
Hence,
\begin{equation}
  \label{eq:15}
 \|\PP_0u\|_2^2 \leq  \Re\sum_{k,m=1}^d\|\partial_{A_m}\partial_{A_k}u\|_2^2 +  C(\|\Delta_Au\|_2^2 + \|Vu\|_2^2 +\|u\|_2^2)\,.
\end{equation}

It is easy to show that if $u\in\Dg_1$ then $\Delta_Au=0$ on $\partial\Omega$: As
 $\partial^2_{A_n}u = 0$ on $\partial\Omega$ for all $1\leq n\leq d$, both $\Delta_Au|_{\partial\Omega}=0$ and
 $\PP_0 u|_{\partial\Omega}=0$ are equivalent to $\partial^2_du|_{\partial\Omega}=0$. Hence, we may
 conclude from \eqref{eq:47} that
 \begin{displaymath}
   \|\Delta_Au\|_2^2 = \Re\sum_{k,m=1}^d\big[\|\partial_{A_m}\partial_{A_k}u\|_2^2 \\- i(\langle
  B_{km}\partial_{A_m}u,\partial_{A_k}u\rangle+\langle
  \partial_{A_m}B_{km}u,\partial_{A_k}u\rangle)\big]\,.
 \end{displaymath}
Hence, by \eqref{eq:15}
\begin{multline*}
   \|\PP_0u\|_2^2 \leq \|\Delta_Au\|_2^2 -
   \Im\sum_{k,m=1}^d\big[(\langle B_{km}\partial_{A_m}u,\partial_{A_k}u\rangle+\langle\partial_{A_m}B_{km}u,\partial_{A_k}u\rangle)\big] \\+  C(\|\Delta_Au\|_2^2 + \|Vu\|_2^2 +\|u\|_2^2)\,.   
\end{multline*}
Using \eqref{eq:10} and \eqref{eq:34} once again yields
\begin{displaymath}
 \|\PP_0u\|_2^2 \leq  C\, (\|\Delta_Au\|_2^2 + \|Vu\|_2^2 +\|u\|_2^2)\,,
\end{displaymath}
and hence
\begin{displaymath}
  \|\PP u\|_2^2 \leq  C\, (\|\Delta_Au\|_2^2 + \|Vu\|_2^2 +\|u\|_2^2)\,.
\end{displaymath}
The proof of \eqref{eq:34} can now be easily completed with the aid
of \eqref{eq:48}(which is valid for all $u\in\Dg$).
\end{proof}

An immediate conclusion is
\begin{corollary}
  Let $\{\mu_n\}_{n=1}^\infty$ and $\{\nu_n\}_{n=1}^\infty$ be defined as in
  Lemma \ref{lem:3.1}, and let $\Omega=\R^d_+$ or $\Omega=\R^d$. Then, there
  exists $C>0$ and $n_0\in\N$ such that
  \begin{equation}
\label{eq:16}
    \mu_n \leq C(1+\nu_n) \,.
  \end{equation}
Thus, if the resolvent of $\PP$ is in $C^p(L^2(\Omega,\C))$ for some
$p>0$, the the same conclusion follows for the resolvent of
$-\Delta_A+|V|$. 
\end{corollary}
\begin{proof}
  Let $u\in\Dg_1$. We need to show first that $(-\Delta_A+|V|)u\in\Dg$, or
  equivalently that $\Delta_Au|_{\partial\Omega}=0$ in the case $\Omega=\R^d_+$ (where
  $\partial\Omega$ is the hyperplane $x_d=0$). Since
  $u\in\Dg$, it follows that $\partial_{A_k}^2u|_{\partial\Omega}=0$ for all
  $1\leq k\leq d-1$. It is easy to show (see \eqref{eq:37}) that
  $\partial_{A_d}^2u|_{\partial\Omega}=0$ using the fact that $\PP_0u|_{\partial\Omega}=0$ for
  all $u\in\Dg_1$.

Since by Lemma \ref{lem:dens} $\Dg$ is dense in $D(-\Delta_A+|V|)$, it
follows that $\Dg_1$ is dense in $D\big((-\Delta_A+|V|)^2\big)$ under the
same graph norm. We can now use statement 2 of Theorem
\ref{LaxMilgramv2} with $S=(-\Delta_A+|V|)^2$ and $\Vg=D(-\Delta_A+|V|)$ to
prove that $\Dg_1$ is dense in $D(-\Delta_A+|V|)$. The proof of
\eqref{eq:16} can now be completed by using Proposition 11.9 in
\cite{hel13} as in the proof of Lemma \ref{lem:opt}.
\end{proof}

\section{Examples}
\label{sec:4}

\subsection{One-dimensional examples}
\subsubsection{The entire real line case.} 
Consider the case $\Omega=\R$ and
\begin{equation}
\label{eq:27}
  \LL = -\frac{d^2}{dx^2} + e^{i\theta}v_1(x)+v_2(x)\,,
\end{equation}
where $\theta\in(-\pi,\pi)$, defined initially on $C_0^\infty(\mathbb R)$.\\
In the above, $v_1(x)\in C^1(\R,\R)$ has  the asymptotic behaviour
\begin{equation}
\label{eq:28}
  |v_1|\sim|x|^\alpha \quad \text{as } |x|\to\infty \,,
\end{equation}
for some $\alpha>0$, 
\begin{equation}
\label{eq:29}
  |v^\prime_1| \leq C(1+|x|^{\alpha-1}) \,,
\end{equation}
and $v_2\in L^\infty_{loc}(\R,\C)$ satisfies
\begin{equation}
\label{eq:30}
  v_2 = o(|x|^\alpha) \quad \text{as } |x|\to\infty \,.
\end{equation}

Set $\PP= i \,\sign \theta \,e^{-i \theta } \LL$. Then $\PP$ meets
\eqref{eq:1}-\eqref{eq:6}, and we may apply Theorems \ref{thm:1.1}
and~\ref{prop:1.2}.  Henry \cite{he13} considers this example for the
case\break  $v(x)=x^{2k}$, $v_2\equiv0$.

By Theorem \ref{thschatten} it follows that $\LL^{-1}$ is in $
C^p$ for all $p>1/2+1/\alpha$. As the numerical range of $\LL+\mu$, for
sufficiently large $\mu\in\R_+$, lies in the sector $\arg \lambda\in[\theta-\pi,\theta]$, Theorem
\ref{thm:1.4} can be applied to obtain, for any $\theta\in( - \pi , \pi )$,
that ${\rm Span}\,(\LL) =L^2(\R,\C)$ whenever $\alpha > 2$.

In the particular case where $v_1 \to +\infty$, as $|x|\to\infty$, the numerical
range lies inside the sector $\arg \lambda\in[0,\theta]$ and hence for any
$\theta\in(-\pi,\pi)$ and $\alpha>\frac{2 |\theta|}{2\pi-|\theta|}$, ${\rm Span}\,(\LL)
=L^2(\R,\C)$.

The complex cubic oscillator is an example included in the
  class we introduce in \eqref{eq:27} (for the general case where
  $v_1$ can change its sign between $-\infty$ and $\infty$) which has been
frequently addressed in the literature (cf. \cite{sh02},
\cite{grma13}, \cite{detr00}, \cite{he13a} to name just a few
references). In this case, $\theta=\pi/2$ and
\begin{displaymath}
  v_1(x)=x^3 \quad,\quad  v_2(x)= \beta_2 x^2+ \beta_1 x \,,
\end{displaymath}
where $\beta_1,\beta_2\in\C$.  In \cite{si75}, the existence of an infinite
sequence of eigenvalues has been established (cf. \cite{sh02} for more
details). Completeness of the system of eigenfunctions in $L^2(\R,\C)$
has been established in \cite{sikr12} for the case $\beta_1=\beta_2=0$. In
\cite{he13a} completeness is extended to the case $i\beta_1\in\R$,
$\beta_2=0$. Here we show it in greater generality, without the need to
rely on the symmetries of the particular cases addressed in
\cite{he13a,sikr12}.

Another case that has been addressed in \cite{da02,da00} is
\begin{equation}
\label{eq:33}
  \LL=-\frac{d^2}{dx^2}+ ce^{i\theta}|x|^\alpha \,,
\end{equation}
where $c>0$ and $\theta\in(-\pi/2,\pi/2)$. It is stated in \cite{da00}, based on
\cite{sh90}, that the eigenspace of $\LL$ is dense in $L^2(\R,\C)$
when either $\alpha\geq1$, or $\alpha\in(0,1)$ and
\begin{displaymath}
  \theta\leq\frac{\pi\alpha}{2} \,.
\end{displaymath}
Since \eqref{eq:33} is a particular case of \eqref{eq:27}, it follows
that whenever $\theta\in(-\pi,\pi)$, $\LL^{-1}$ must be in $C^p$ for all
$p>1/2+1/\alpha$. The numerical range of $\LL$ however, is confined in
$\arg z\in(0,\theta)$. Hence, the eigenfunctions of $\LL$ form a complete
system whenever
\begin{displaymath}
  |\theta|< \frac{2\pi\alpha}{\alpha+2} \,.
\end{displaymath}
In particular, if $\theta\in(-\pi/2,\pi/2)$ completeness of the eigensystem
of $\LL$ is guaranteed for all $\alpha>2/3$. Note that for $0<\alpha<1$
\begin{displaymath}
  \frac{\pi\alpha}{2} < \frac{2\pi\alpha}{\alpha+2}\,.
\end{displaymath}
Hence, our method provides greater domains for $\theta$ and $\alpha$ where the
eigenspace of $\LL$ is dense in $L^2(\R,\C)$.

\subsubsection{The positive real line case:} 

Here we consider the same differential operator as in \eqref{eq:27},
but this time defined on $C_0^\infty(\mathbb R^+)$ and consider the
Dirichlet realization.  Assuming that $v_1$ and $v_2$ satisfy
\eqref{eq:28}, \eqref{eq:29}, and \eqref{eq:30}, we have by Theorem
\ref{thschatten} that $\LL^{-1}\in C^p$ whenever $p>1/2+1/\alpha$ as
before. However, as in the case where $\LL$ is given bt \eqref{eq:33},
since every direction outside $[0,\theta]$ is a direction of minimal
growth for $(\LL_+-\lambda)^{-1}$. Thus, we have that ${\rm Span}\,\LL_+
=L^2(\R,\C)$ whenever
\begin{displaymath}
  \theta < \frac{2\alpha\pi}{\alpha+2}\,.
\end{displaymath}
In particular, for the case $v_1=x$ and $v_2=0$, which is known as the
complex Airy's equation, we obtain that ${\rm Span}\,\LL_+
=L^2(\R,\C)$ whenever $\theta<2\pi/3$.

\subsection{Two-dimensional examples}

\subsubsection{Electric and magnetic potentials.} Let $A$ and $V$ be such that
$V=i\phi$ and $\curl A +i\phi$ is a holomorphic function of $z=x+iy$, $x$
and $y$ being the planar coordinates. Such a choice is in line with
the steady state Faraday's law, which in two dimensions read (assuming
all constants are equal to $1$)
\begin{displaymath}
  \nabla_\perp\curl A + \nabla\phi =0 \,.
\end{displaymath}
We further narrow our choice by setting
\begin{equation}
\label{eq:31}
  \curl A +i\phi = z^n \,.
\end{equation}
Consider the differential operator
$-\Delta_A+i\phi$ with $\Omega=\R^2$ (which is clearly a particular case of
\eqref{eq:6}). By  Theorem \ref{thschatten}, we 
have that $\PP^{-1}\in C^p$ for all $p>1+2/n$. It can be readily
verified that, independently of $\beta$, every direction outside
$[-\pi/2,\pi/2]$ is a direction of minimal growth for
$(\PP-\lambda)^{-1}$. Unfortunately, the condition of validity of our
theorems will lead to the condition $1 < \frac{1}{1 + \frac{2}{n}}$,
which is never satisfied. However, if we consider instead of $\R^2$, a
smooth domain $\Omega$ which is contained in a sector whose opening is
smaller than $\pi/(n+2)$, and such that the positive real axis is
contained in this sector, then the numerical range of $\PP$ is
contained in a sector whose opening is less than $\pi/(1+2/n)$ and we
can use Theorem \ref{thm:1.4} to prove completeness of the
eigenspace.

\subsubsection{Analytically dilated operators}
Next, for $m\geq2$ and $k\geq1$, let
\begin{equation}\label{model1}
  \A = -\frac{\partial^2}{\partial x^2} - \Big(\frac{\partial}{\partial y}-i\frac{x^m}{m}\Big)^2
  +iy^{2k} \,,
\end{equation}
Note that the numerical range lies in the sector $\arg \lambda\in[0,\pi/2]$.
In this case we use analytic dilation, as in \cite{AHP2}, to reduce
the angle of the sector where the numerical range lies. This would
allow us to establish the existence of countable set of eigenvalues,.
Thus, we define the dilation operator
\begin{displaymath}
u\to(\Ug u)(x,y)=e^{-i\frac{m-1}{2}\alpha}u(e^{i\alpha}x,e^{-im\alpha}y)
\end{displaymath}
where
\begin{equation}
\label{eq:32}
  \alpha=-\frac{\pi}{4m(k+1)} \,,
\end{equation}
to obtain that
\begin{equation} \label{modele2}
  \A_\alpha=\Ug^{-1}\A\Ug=-e^{2i\alpha}\frac{\partial^2}{\partial x^2} -
  e^{-2im\alpha}\Big(\frac{\partial}{\partial y}-i\frac{x^m}{m}\Big)^2 +e^{2ikm\alpha  +i \pi/2}y^{2k} \,.
\end{equation}

We next verify $\A_\alpha$ meets the conditions Theorem \ref{thm:1.1}.
Since $V_1=e^{2ikm\alpha+i \pi/2}y^{2k}$ and $$ 2km \alpha + \frac \pi 2=
-\frac{k \pi}{2(k+1)} + \frac \pi 2= \frac{ \pi}{2(k+1)} \,,$$ we obtain
that $$\Re V_1 = \cos \frac{ \pi}{2(k+1)} y^{2k} \,.$$ Hence the
conditions of Theorem \ref{thm:1.1} can be easily verified. We thus
obtain that $(\A_\alpha)^{-1}$ is in $C^p$ for all
\begin{equation}
\label{eq:50}
  p > 1+ \frac{1}{2k} + \frac{1}{m-1}=\frac{(2k+1)m-1}{2k(m-1)} \,.
\end{equation}

Next we observe that
\begin{displaymath}
  \langle u,(e^{-2i\alpha}\A_\alpha-\lambda)u\rangle = \|u_x\|_2^2 +
  e^{-2i(m+1)\alpha}\Big(\Big\|\Big(\partial_y-i\frac{x^m}{m}\Big)u\Big\|_2^2 +
  \|y^ku\|_2^2\Big) - \lambda\|u\|_2^2 \,.
\end{displaymath}
It can be easily verified from the above and \eqref{eq:32} that the
numerical range of $\A_\alpha$ lies in the sector
\begin{displaymath}
  \arg \lambda \in \Big[0,\frac{(m+1)\pi}{2m(k+1)}\Big] \,.
\end{displaymath}
which is always contained in $[0,3\pi/8]$ for the assumed range of $m$
and $k$ values.  Consequently, the effect of this analytic dilation
has permitted us to reduce the angle of the sector in which the
numerical range lies. Since
\begin{equation}
\label{eq:49}
  \frac{(m+1)}{2m(k+1)} < \frac{2k(m-1)}{(2k+1)m-1}\,,
\end{equation}
we obtain that ${\rm Span}\,\A_\alpha =L^2(\R^2,\C)$. By the same
arguments of \cite{AHP2} we then obtain that $\A$ has an infinite
sequence of eigenvalues for all $k\geq1$ and $m\geq2$.
\begin{remark}~
\begin{itemize}
\item It is not clear whether analytic dilation preserves the
  completeness of the system of generalized eigenfunctions.  Hence,
  the best we can obtain is the existence of an infinite discrete
  spectrum for $\A$.
\item Had we abandoned analytic dilation, we would have obtained,
  instead of \eqref{eq:49}, the condition
$$
\frac 12 < \frac{2k(m-1)}{(2k+1) m -1}
$$
which implies $m>2$ and 
\begin{equation}
k > \frac{m-1}{2 (m-2)}
\end{equation}
to achieve completeness.
\item Note that the results presented in \cite{L-BNo} can be applied
  to \eqref{modele2} to conclude the optimality of \eqref{eq:50}, as
  they can be used to obtain the precise Schatten class of the
  resolvent of $-\Delta_A + y^{2k}$ in $\R^2$. 
\end{itemize}
\end{remark}

\subsubsection{A half-plane problem.} 
Let
\begin{displaymath}
  \R^2_+ = \{ (x,y) \in \R^2 \,| \, y>0\}\,.
\end{displaymath}
Consider the case where $V=e^{i\theta}y$ and $A=x^2\hat{i}_y/2$.  We define
$\PP_+=-\Delta_A+V$ with $\Omega=\R^2_+$ (which is once again a particular case
of \eqref{eq:6}). Once again we have that $\PP_+^{-1}\in
C^p$ for all $p>1+2/n$. For $n=1$ and $\beta=0$ it can be easily shown
that the numerical range of $\PP_+$ is confined within the sector
$[0,\theta]$ in $\C$. Hence, every direction outside $[0,\theta]$ is a direction
of minimal growth for $(\PP_+-\lambda)^{-1}$, and consequently, ${\rm
  Span}\,\PP_+ =L^2(\R^2_+,\C)$ for all $\theta<\pi/3$.

\paragraph{Acknowledgements}
Y. Almog was supported by NSF grant DMS-1109030 and by US-Israel BSF grant
no.~2010194. 


\end{document}